\documentclass[letter]{article}

\usepackage[ruled,vlined,linesnumbered]{algorithm2e}
\usepackage[fleqn]{amsmath}

\usepackage{amsthm}
\newcommand{\nosemic}{\SetEndCharOfAlgoLine{\relax}}
\usepackage{amssymb}

\usepackage{physics}
\usepackage{graphicx}

\newtheorem{theorem}{Theorem}[section]
\newtheorem{corollary}{Corollary}[theorem]

\theoremstyle{remark}
\newtheorem*{property}{Property}

\theoremstyle{remark}

\newcommand{\state}{x}
\newcommand{\stateRV}{X}
\newcommand{\action}{a}
\newcommand{\actionRV}{A}
\newcommand{\attime}[2]{#1_{#2}}
\newcommand{\attimes}[2]{#1_{[#2]}}
\newcommand{\timedstates}[1]{\attimes{\state}{#1}}
\newcommand{\timedstate}[1]{\attime{\state}{#1}}
\newcommand{\timedactions}[2][\state]{\attimes{\action}{#2}}
\newcommand{\timedaction}[2][\state]{\attime{\action}{#2}}

\newcommand{\timedStates}[1]{\attimes{\stateRV}{#1}}

\newcommand{\timedActions}[2][\state]{\attimes{\actionRV}{#2}}
\newcommand{\timedAction}[2][\state]{\attime{\actionRV}{#2}}
\newcommand{\softmax}{\text{soft}\hspace{-0.3em}\max}

\title{Maximum Likelihood Constraint Inference from Stochastic Demonstrations}

\author{David L. McPherson, Kaylene C. Stocking, S. Shankar Sastry}

\date{\today}

\begin{document}
\maketitle

\begin{abstract}
  When an expert operates a perilous dynamic system, ideal constraint information is tacitly contained in their demonstrated trajectories and controls.
The likelihood of these demonstrations can be computed, given the system dynamics and task objective, and the maximum likelihood constraints can be identified.
Prior constraint inference work has focused mainly on deterministic models.
Stochastic models, however, can capture the uncertainty and risk tolerance that are often present in real systems of interest.

This paper extends maximum likelihood constraint inference to stochastic applications by using maximum causal entropy likelihoods.
Furthermore, we propose an efficient algorithm that computes constraint likelihood and risk tolerance in a unified Bellman backup, allowing us to generalize to stochastic systems without increasing computational complexity.


\end{abstract}

\section{Introduction}
Optimization-based control (notably, model predictive control) promises intelligent behavior \cite{aswani2013provably} even in challenging nonlinear dynamics \cite{ng2003autonomous} and even stochastic dynamics \cite{williams2016aggressive}\cite{broek2012risk}.
It's already impacted industrial practice for decades \cite{qin2003survey} through model-predictive control, and its recent incarnation as ``deep reinforcement learning'' \cite{ng2003autonomous}\cite{sunderhauf2018limits} has been promising to revolutionize control again.

Yet these optimizations must always begin with a fundamental question: what should the automation optimize?
Answering this question is notoriously difficult \cite{amodei2016concrete} and lies at the core of the so-called "Value Alignment" problem of AI-safety.

One solution approach is to solve the inverse of optimal control: given near-optimal demonstrations, recover the metric the demonstrator is optimizing \cite{kalmanInverseOptimalControl}
After fitting the task specification the objective can be optimized to imitate the expert behavior \cite{abbeel2004apprenticeship}.
or used to predict human motion \cite{ziebart2009human}.

Often, inverse optimal control focuses on inferring the objective metric that is being continuously optimized.
However, in parallel to constrained optimal control for safety, algorithms are being extended to infer those hard constraints \cite{armesto2017efficient}\cite{chou2020learning}\cite{li2016learning}\cite{perez2017c}\cite{scobee2019maximum}.
Chou et al. \cite{chou2018learning} inferred constraints along the paths that would be low cost but were never observed.
This intuition was formalized by Scobee et al. \cite{scobee2019maximum} by translating maximum entropy inverse reinforcement learning \cite{ziebart2008maximum} to work for hard constraints.
Unfortunately, the maximum entropy used in \cite{scobee2019maximum} only works for deterministic systems.


Non-deterministic dynamics reflect the uncertainty that exacerbates the safety risks beyond merely avoiding bad controls to actively combating diffusion into unsafety.
Stochasticity can model a variety of unpredictable dynamics in applications: from unpredictable power sources in renewable power systems \cite{khodayar2013enhancing} to
hard-to-model turbulence in road conditions \cite{williams2016aggressive}, from tumor cell growth in cancer treatment \cite{risom2018differentiation} to unforeseen changes in stormwater reservoirs \cite{chapman2019risk}.

The maximum entropy likelihoods can be extended to uncertain transition dynamics by conditioning the entropy at each time step only on the previously revealed state transitions
\cite{ziebart2010modeling}.
This maximum causal entropy has been applied to learn general non-Markovian specifications \cite{vazquez2018learning}.
A subclass of non-Markovian specifications are state-action constraints that proscribe that a configuration or control should never be taken.
This paper specializes the inference to just this sub-class paralleling the approach of \cite{scobee2019maximum} but applying to causal entropy in stochastic dynamics.


\subsection{Contributions and Guide}
This work advances prior art \cite{scobee2019maximum} in inferring state-action constraints:

\begin{itemize}
  \item by respecting causality using the principle of maximum \textbf{causal} entropy for likelihood generative models
  \item by extending the hypothesis family to include risk-tolerating \textbf{chance} constraints
  \item and by streamlining the algorithm into one backwards pass, thereby maintaining the same computational complexity as the non-stochastic version \cite{scobee2019maximum}
\end{itemize}

\section{Background}
\subsection{Maximum Likelihood Infererence}
The inference task is to fit parameters to observed demonstration trajectories.
We follow the maximum likelihood inference framework that models data as samples of a parametrized distribution.
In our case, the distribution can be derived from the task description  and agent model. These will be described in subsections \ref{sec:MDP} and \ref{sec:Boltzmann}, respectively.
The parameters to fit will be which candidate constraints $C \in \mathcal{C}$ are effecting the expert and at what risk tolerance $\alpha_C$. These chance constraints will be described in subsection \ref{sec:ChanceConstraints}.

\subsection{Markov Decision Processes}
\label{sec:MDP}
In general, a stochastic process is any indexed collection of random variables.
For discrete time processes, these variables are indexed by integers $t \in [0:T]$.
A discrete-time Markov process $X_t \in \mathcal{X}, \quad t \in [0:T]$ is one with the discrete-time Markov property:

\begin{align}
  P(X_t | X_{t-1}, X_{t-2}, \dots, X_0) = P(X_t | X_{t-1}), \forall t>0
\end{align}

The Markov property implies that conditioning on an earlier indexed state contains all the past information before it that is relevant to determining the current state.

A Markov decision process parametrizes these Markovian transitions by a series of chosen actions $a_t \in \mathcal{A}$.
This conditional transition distribution can therefore (assuming that the state space $\mathcal{X}$ is a measurable space) be captured by a probability density function parametrized by previous state $x_{t-1}$ and chosen action $a_{t-1}$:

\begin{align}
  P_{a_{t-1}}(X_t = x_t | X_{t-1} = x_{t-1}) = S(x_{t-1}, a_{t-1}, x_{t})
  \label{eq:MDPtransitionDef}
\end{align}

These actions are then evaluated by a metric on the state and action sequences into some scalar reward or penalty, $R(x_{0:T},a_{0:T-1})$:

\begin{align}
  R(\timedstates{0:T},\timedactions[\cdot]{0:T-1}) = \sum_{t=0}^{T-1} r(x_{t},a_{t}) + w(x_{T})
  \label{eq:mdp_reward}
\end{align}
where $r$ is the running cost and $w$ is the final cost.

Bundling these four objects of the:
\begin{itemize}
  \item state space $\mathcal{X}$,
  \item set of actions $\mathcal{A}$,
  \item transition distribution function $P_{a_{t-1}}(x_t | x_{t-1})$,
  \item and objective metric $R(\timedstates{0:T},\timedactions[\cdot]{0:T-1})$
\end{itemize}
makes up a 4-tuple that defines the Markov Decision Process.

This process' distribution on reward outcomes induced by certain action choices becomes the foundation of stochastic control problems in discrete spaces.
This paper models the expert's performance as operating in this discrete space under additional constraints. We describe these additional constraints next.

\subsection{Chance Constraints}
\label{sec:ChanceConstraints}
The agent must also choose its actions to avoid dangerous states $x \in C_X$.
To model some risk-tolerance, we allow some small probability $\psi(x)$ of transitioning to an $x \in C_X$:

\begin{align}
  P(X_{t+1} = x | X_t = x_t, a_t) \leq \psi(x), \quad \forall x \in C_X
  \label{eq:transitionChanceConstraint}
\end{align}

To deterministically constrain out a state $x$ set $\psi(x) = 0$.
On the other hand, setting $\psi(x) = 1$ means the constraint is inactive and transitioning to $x$ is freely allowed.
Therefore the set of state constraints $C_X$ can be encoded as a $\psi(x)$ over all states $x \in \mathcal{X}$.

There can also be constraints on action that rule out illegal actions $a \in C_A$. Since the stochastic dynamics only make states uncertain, we only need chance constraints on the states and not on the actions.

Let $C$ be the tuple $(\psi(x),C_A)$, and call the set of all these constraint candidates $\mathcal{C}$, so $C \in \mathcal{C}$.

Let the set $W_{C}^{t}(x_t)$ be the set of actions $a_t$ that satisfy $C_A$ from $x_t$ and generate state transitions that satisfy Eq. (\ref{eq:transitionChanceConstraint}).
And let $\Phi_C^t(a,x)$ be its indicator: the indicator of whether Eq. (\ref{eq:transitionChanceConstraint}) and $a \notin C_A$ are satisfied:

\begin{align}
  \Phi_C^t&(a,x) = \nonumber \\
  = \mathbb{I}
    [
      &a \notin C_A \nonumber\\
      &\cap
      (P(X_{t+1} = \hat{x} | X_t = x, a) \leq \psi(\hat{x}) \forall \hat{x} \in C_X)
    ]
    \label{eq:chanceConstraintIndicatorDef}
\end{align}

\subsection{Boltzmann Distribution on Controller Sequences}
\label{sec:Boltzmann}
All the probabilities in the previous section were based purely on process noise that made $X_t$ random given a deterministic choice of controllers $\timedaction[\cdot]{t}$ for $t \in [0,T]$.
The chance constraints outlaw certain induced distributions over $X_t$ given $\timedaction[\cdot]{t-1}$.

Inside these chance constraints, the agent chooses actions.
There are many possible ways an agent might choose their actions.
We assume they are endeavouring to optimize the reward function $\sum_t r(X_t,a_t) + w(X_T)$ defined in Eq. (\ref{eq:mdp_reward}) and assume nothing else.
The Maximum Entropy method optimizes the generative model distribution $P_C(a | x)$ to fit this known constraint-set $C$ and leave all other facets maximally agnostic \cite{ziebart2010modeling}.
This distribution layers on another layer of stochasticity purely on the level of selecting controller sequences $\timedAction[\cdot]{t} \quad \forall t \in [0:T]$.
Ziebart \cite[p.~74]{ziebart2010modeling} finds the maximal causal entropy distribution to be compactly defined by a backwards iteration with close analogues to the Bellman backup only with the $\max$ exchanged for a differentiable approximation $\softmax$:

\begin{align}
  P_C(a_t | x_t)
  &= \frac{
      e^{
        Q_{C,t}^{soft}(
          \timedaction[\cdot]{t},
          \attime{x}{t}
        )
      }
    }{
      e^{
        V_{C,t}^{soft}(
          \attime{x}{t}
        )
      }
    }
    \Phi_C^t(a,x)
    \label{eq:PBellmanDef}
  \\
  Q_{C,t}^{soft}(\timedaction[\cdot]{t}, \attime{x}{t})
    &= r(x_t,a_t) 
       + \mathbb{E}_{X_{t+1}} V_{C,t+1}^{soft}(\attime{x}{t+1}) \label{eq:Qdef}
  \\
  V_{C,t}^{soft}(\attime{x}{t})
    &= \log \sum_{a_t \in W_C^t}
      e^{
        Q_{C,t}^{soft}(
          \timedaction[\cdot]{t},
          \attime{x}{t}
        )
      } \label{eq:Vdef}
    \\
    &= \softmax_{a_t \in W_C^t} Q_{C,t}^{soft}(\timedaction[\cdot]{t}, \attime{x}{t})
\end{align}
where $Q^{soft}$ can be interpreted as a state-action soft-optimal value-to-go and $V^{soft}$ the state's soft-optimal value-to-go.

Note that these parallels are coincindental as these recursions actually derive from tracking normalizing constants on causally normalized distributions.

These ``soft Bellman'' distributions in Equation (\ref{eq:PBellmanDef}) over single-timestep actions $P(a_t | x_t)$ can form a joint distribution over horizon-wide sequences of controllers:

\begin{align}
  P_C(&\timedActions[\cdot]{t:T} = \timedactions[\cdot]{t:T} | X_t = x_t)
  \\
  &=
  \begin{cases}
    \frac{
      e^{
        \mathbb{E}
        \left[
          R(
            \timedStates{t:T},
            \timedactions[\cdot]{t:T}
          )
        \right]
      }
    }{
      e^{
        V_{C,t}^{soft}(
          \timedstate{t}
        )
      }
    }
    &, \text{ if } \timedactions[\cdot]{t:T} \in W_C^{[t:T]}\\
    0 &, \text{ if } \timedactions[\cdot]{t:T} \notin W_C^{[t:T]}
  \end{cases} \label{eq:BoltzmannDef}
\end{align}

\section{Constraint Inference Statistics}
This full-horizon-wide formulation of the action distribution is useful for relating this generative behavior model to the empirical expert demonstrations.
Let the cross product of the $W_{C^+}^{t}(\cdot)$ for all time points will be $W_{C^+}^{[0:T]}$: the set of controller sequences $a_{0:T}(\cdot)$ that satisfy $C_A$ for all possible input states and generate state transitions that satisfy Eq. (\ref{eq:transitionChanceConstraint}) for all time.
The demonstrations must be within $W_{C^+}^{[0:T]}$.
So when inferring constraints, we can instantly rule out any constraint-sets $C^+$ that place the demonstrations outside of its $W_{C^+}^{[0:T]}$.
Amongst the remaining constraint-set candidates $C^+ \in \mathcal{C}$ we choose the one that maximizes the likelihood of observing the demonstrations.
Indeed, across all demonstrations each likelihood rescales by a $F_{C^+}(x_t)$ unique to each candidate constraint:

\begin{align}
  P_{C^+}(&
    \timedActions[\cdot]{t:T}
    = \timedactions[\cdot]{t:T}
    | X_t = x_t
  )
  \\ &=
  \frac{
    e^{
      \mathbb{E}
      \left[
        R(
          \timedStates{t:T},
          \timedactions[\cdot]{t:T}
        )
      \right]
    }
  }{
    e^{
      V_{C^+,t}^{soft}(
        \timedstate{t}
      )
    }
  }
  \\ &=
  \frac{
    e^{
      \mathbb{E}
      \left[
        R(
          \timedStates{t:T},
          \timedactions[\cdot]{t:T}
        )
      \right]
    }
  }{
    e^{
      V_{C^0,t}^{soft}(
        \timedstate{t}
      )
    }
  }
  \frac{
    e^{
      V_{C^0,t}^{soft}(
        \timedstate{t}
      )
    }
  }{
    e^{
      V_{C^+,t}^{soft}(
        \timedstate{t}
      )
    }
  }
  \\ &=
  P_{C^+}(
    \timedActions[\cdot]{t:T}
    = \timedactions[\cdot]{t:T}
    | X_t = x_t
  )
  \frac{1}{F_{C^+,t}(x_t)}
\end{align}
where we've set $F_{C^+,t}(x_t)$ to:

\begin{align}
  F_{C^+,t}(x_t) =
  \frac{
    e^{
      V_{C^+,t}^{soft}(
        \timedstate{t}
      )
    }
  }{
    e^{
      V_{C^0,t}^{soft}(
        \timedstate{t}
      )
    }
  } \label{eq:Fdef}
\end{align}

Therefore the likelihoods can be computed for just one constrained optimal control problem, call it $C^0$, and then readily translated into the likelihoods for all constraints under consideration.
Prior art \cite{scobee2019maximum} bounded the sub-optimality of inferring the constraints incrementally by adding one state to $C_X$ or action to $C_A$ per step.
Amongst all these candidate $C^+$, whichever has the lowest $F_{C^+,0}(x_0)$ (assessed at the start $x_0$) will have the highest likelihood for all the demonstrations and be the maximum likelihood constraint.

\begin{property}
  Consider two candidate constraints $C^+$ and $C^+_{--}$
  that differ only by $C^+_{--}$ having exactly one $\psi(x)$ lower than $C^+$ has.

  $C^+_{--}$ will always have $F_{C^+_{--},0}(x_0) \leq F_{C^+,0}(x_0)$.
\end{property}

\begin{corollary}
  When considering adding a single state $x$ into $C_X$ always choose the lowest possible $\psi(x)$ that doesn't rule out any demonstrations.
  \label{cor:chanceConstraintLevels}
\end{corollary}



\section{Dynamic Programming}
\label{sec:DynDeriv}

The ratio defined in Equation (\ref{eq:Fdef}) can be computed by modifying the soft Bellman backup defined in Equations (\ref{eq:PBellmanDef}) - (\ref{eq:Vdef}).
This modified backup procedure is described in the below theorem:

\begin{theorem}
  \label{thm:dynprog}
  Let $C^0$ be a set of constraints and $C^+$ be an augmented version of $C^0$ so that $W_{C^+} \subset W_{C^0}$. Then any $F_{C^+,t}(x_t)$ can be computed with the same sums as for the base $C^0$

  \begin{align*}
    F&_{C^+,t}(x_t)
    \\&=
    \mathbb{E}_{
      a_t \sim P_{C^0}
    }
    \left[
      \Phi_{C^+}^t(
        a_t,
        x_t
      )
      e^{
        \mathbb{E}_{
          x_{t+1}
        }
        \log(
          F_{C^+,t+1}(x_{t+1})
        )
      }
    \right]
  \end{align*}
\end{theorem}

\begin{proof}

  \begin{align*}
    F_{C^+,t}(x_t)
    &=
    \frac{
      e^{
        V_{C^+,t}^{soft}(
          x_{t}
        )
      }
    }{
      e^{
        V_{C^0,t}^{soft}(
          x_{t}
        )
      }
    }
    \\
    &=
    \frac{
      \sum_{a_t \in W_{C^+}^t}
        e^{
          Q_{C,t}^{soft}(
            a_{t},
            x_{t}
          )
        }
    }{
      e^{
        V_{C^0,t}^{soft}(
          x_{t}
        )
      }
    }
    \\
    &=
    \sum_{a_t \in W_{C^+}^t}
      \frac{
        e^{
          Q_{C,t}^{soft}(
            a_{t},
            x_{t}
          )
        }
      }{
        e^{
          V_{C^0,t}^{soft}(
            x_{t}
          )
        }
      }
    \\
    &=
    \sum_{a_t \in W_{C^+}^t}
      \frac{
        e^{
          r(x_t,a_t) +
          \mathbb{E}_{
            x_{t+1}
          }
          V_{C^+,t+1}^{soft}(x_{t+1})
        }
      }{
        e^{
          V_{C^0,t}^{soft}(
            x_{t}
          )
        }
      }
  \end{align*}

  It will be convenient to define the logarithm of our $F_{C,t}$. Let it be $\Delta_{C}^t$:

  \begin{align*}
    \Delta_{C^+}^{t+1}(
      x_{t+1}
    )
    &=
    \log(
      F_{C^+,t+1}(x_{t+1})
    )
    \\
    &=
    \log(
      \frac{
        e^{
          V_{C^+,t+1}^{soft}(
            \timedstate{t+1}
          )
        }
      }{
        e^{
          V_{C^0,t+1}^{soft}(
            \timedstate{t+1}
          )
        }
      }
    )
    \\
    &=
    \log(
      e^{
        V_{C^+,t+1}^{soft}(
          \timedstate{t+1}
        ) -
        V_{C^0,t+1}^{soft}(
          \timedstate{t+1}
        )
      }
    )
    \\
    &=
    V_{C^+,t+1}^{soft}(
      \timedstate{t+1}
    ) -
    V_{C^0,t+1}^{soft}(
      \timedstate{t+1}
    )
  \end{align*}

  Then the ratio can be redefined in terms of previously calculated terms on $C^0$ and our iterating $F_{C,t}$

  \begin{align}
    F_{C^+,t}(x_t)
    &=
    \sum_{a_t \in W_{C^+}^t}
      \frac{
        e^{
          r(x_t,a_t) +
          \mathbb{E}_{
            x_{t+1}
          }
          V_{C^0,t+1}^{soft}(x_{t+1})
        }
      }{
        e^{
          V_{C^0,t}^{soft}(
            x_{t}
          )
        }
      } \nonumber \\ &\hspace{7em}\cdot
      e^{
        \mathbb{E}_{
          x_{t+1}
        }
        \Delta_{C^+}^{t+1}(x_{t+1})
      } \nonumber
    \\ \nonumber
    &=
    \frac{
      \sum_{a_t \in W_{C^+}^t}
        e^{
          Q(x_t,a_t) +
          \mathbb{E}_{
            x_{t+1}
          }
          \Delta_{C^+}^{t+1}(x_{t+1})
        }
    }{
      e^{
        V_{C^0,t}^{soft}(
          x_{t}
        )
      }
    }
    \\ \nonumber
    &=
    \sum_{a_t \in W_{C^+}^t}
      \frac{
        e^{
          Q(x_t,a_t)
        }
      }{
        e^{
          V_{C^0,t}^{soft}(
            x_{t}
          )
        }
      }
      e^{
        \mathbb{E}_{
          x_{t+1}
        }
        \Delta_{C^+}^{t+1}(x_{t+1})
      }
      \\ \nonumber
      &=
      \sum_{a_t \in W_{C^+}^t}
        P_{C^0}(a_t | x_t)
        e^{
          \mathbb{E}_{
            x_{t+1}
          }
          \Delta_{C^+}^{t+1}(x_{t+1})
        }
      \\
      &=
      \mathbb{E}_{a_t \sim P_{C^0}}
        \Phi_{C^+}^t(
          a_t,
          x_t
        )
        e^{
          \mathbb{E}_{
            x_{t+1}
          }
          \Delta_{C^+}^{t+1}(x_{t+1})
        }
  \end{align}

\end{proof}

\section{Algorithm}
Theorem \ref{thm:dynprog} implies that an algorithm can compute the conversion ratios $F_C(x)$ (which will correspond to how much the distribution shrunk by) for all candidate constraints at the same time as the Bellman backup for the baseline set of constraints $C^0$.
The Greedy Iterative Constraint Inference procedure pioneered in \cite{scobee2019maximum} suggests this selection can be performed iteratively adding just one constraint at a time. This iterative approach can be shown to be bounded sub-optimal compared to selecting all the constraints simultaneously \cite{scobee2019maximum}.
In this iterative approach, the $F_{C^+}$ optimizing $C^+$ will become the baseline set of constraints for the next iteraiton $C^i$.

\subsection{Determining Chance Constraint Risk Levels from Demonstrated Transitions}
Corollary \ref{cor:chanceConstraintLevels} states that this set of candidates can be further reduced to only those whose newly added $\psi(x)$ are as exclusive as possible without excluding any of the demonstrations $(\tilde{x}_{0:T},\tilde{a}_{0:T}) \in \mathcal{D}$. That is, when adding state constraints, the newly added exclusion threshold $\psi(x)$ must as low as possible while still being greater than all transition probabilities to $x$ that were chosen by the expert in their demonstrations.
For simplicity, we will bound lowerbound $\psi(x)$ to prevent any precursor states of $x$ from having all its available actions ruled out thereby dooming any trajectory entering that precursor state to necessarily violate the chance constraint on $x$. Therefore this lowerbound $\Psi(x)$ must be defined:

\begin{align}
  \Psi(x') =
  \max_{ \{
    x \mid
      \exists \hat{a}
      \ni
        P(x' |
          x,
          \hat{a}
        )
        > 0
  \} }
    \min_{a}
      P(x' | x, a)
\end{align}

This implies that the new $\psi(x)$ should be:

\begin{align}
  \psi(x) =
  \max \left\{
    \max_{
      (
        \tilde{x}_{0:T},
        \tilde{a}_{0:T}
      )
      \in \mathcal{D}
    }
    \max_{
      t
      \in [0:T-1]
    }
    S(
      \tilde{x}(t),
      \tilde{a}(t),
      x
    ) ,
    \Psi(x)
  \right\}
  \label{eq:riskLevelSelection}
\end{align}

\begin{algorithm}
  \SetAlgoLined
  \KwResult{$V_{C^0,t}$ and a column vector $F$ where each entry corresponds to the $F_{C^+,0}$ that adds one state/action constraint on top of $C^+ \in \mathcal{C^+_0}$}
  \nosemic
    \For{$x \in \mathcal{X}$}{
      $Z(T,x) \leftarrow \exp(w(x))$\;
      $F(T,x) \leftarrow 1$\;
    }
    \For{$t \in [T-1,0]$}{
      \For{$x \in \mathcal{X}$}{
        $Z(t,x) \leftarrow 0$ \;
        $F(t,x) \leftarrow 0$ \;
        \For{$a \in \mathcal{A}$}{
          $Q(t,x,a) \leftarrow r(x,a)$ \;
          $D(t,x,a) \leftarrow 0$ \;
          \For{$x' \in \mathcal{X}$}{
            $Q(t,x,a) += S(x,a,x') \log(Z(t+1,x'))$ \;
            $D(t,x,a) += S(x,a,x') \log(F(t+1,x'))$ \;
          }
          $Z(t,x) += \Phi_{C^0}(x,a) \exp(Q(t,x,a))$ \;
          $F(t,x) += \Phi_{C \in \mathcal{C}}(x,a)$
                    $\exp(Q(t,x,a))$
                    $\exp(D(t,x,a))$ \;
        }
        $F(t,x) =  F(t,x) / Z(t,x)$
      }
    }
  \caption{\label{alg:Fratio}Modified Bellman Backup with Value Ratio}
\end{algorithm}

The most likely constraint is then whichever one still allows the observed demonstrations and has the smallest normalizing constant from the starting state $F_{C,0}(x_0)$.

Note that Algorithm \ref{alg:Fratio} has computations on the order of $O(|\mathcal{X}|^2 (|\mathcal{X}| + |\mathcal{A}|)$, identical to the computational complexity of prior art in maximum likelihood constraint inference \cite{scobee2019maximum}.

\section{Results}

\begin{figure*}
  \centering
  \includegraphics[width=\textwidth]{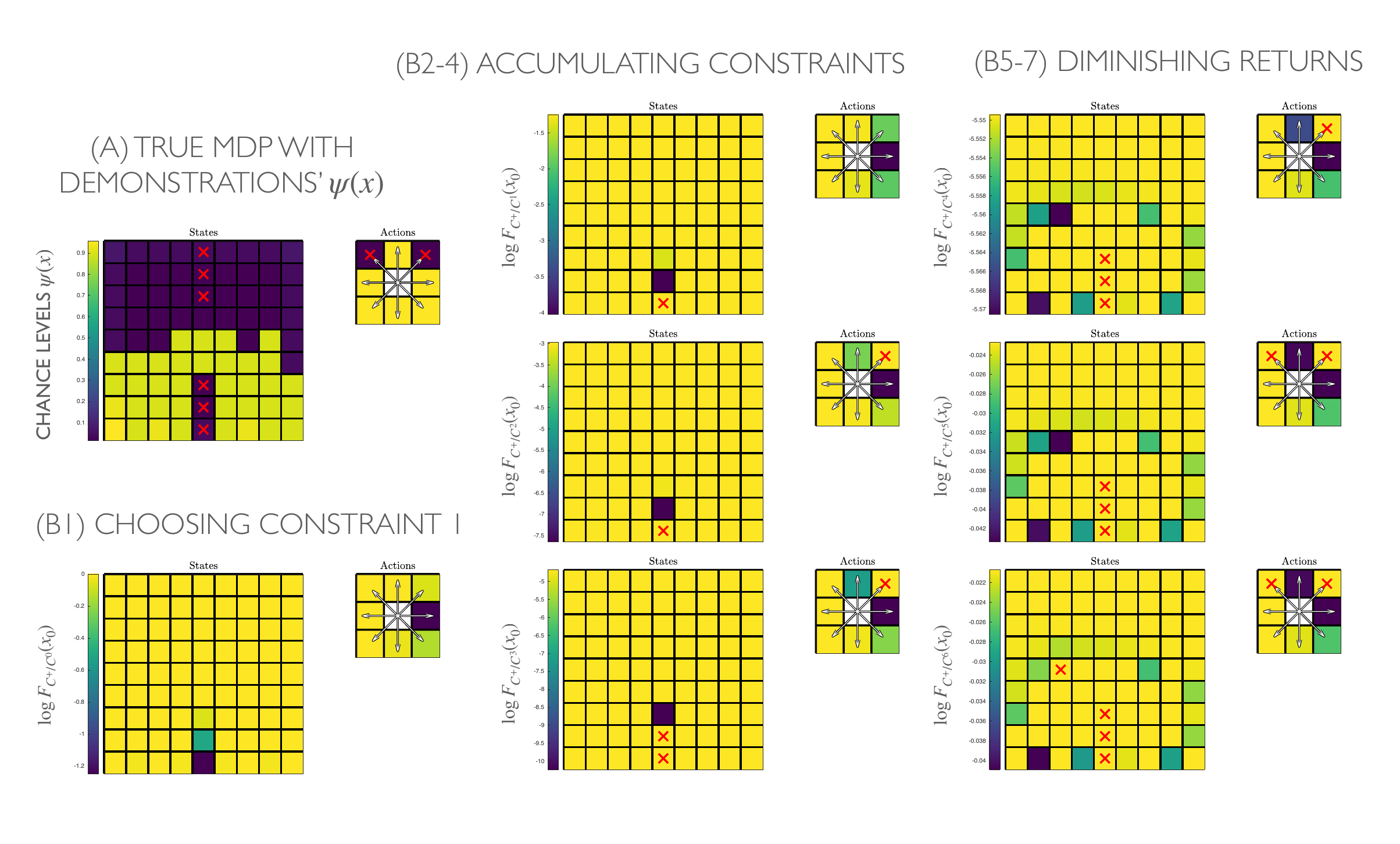}
\caption{\label{fig:demonstration} The constraint inference algorithm was evaluated on a gridworld synthetic dataset with stochastic dynamics. In the panels, X marks represent constraints.
(A) The demonstrator's true set of constraints to be inferred. Cell shading indicates whether that state can be considered as a constraint given the demonstration set. State cell shading represents the $\psi(x)$ allowed by Equation (\ref{eq:riskLevelSelection}). Action cell shading indicates the binary indicator of whether that action was sampled in any of the demonstrations.
(B)
The sequence of inferred constraints alongside the value of adding other candidate constraints in addition.
Cell shading corresponds to how much further partition mass would be eliminated by introducing a constraint on that action or chance constraint on that state (taking the log of the ratio for distinguishability). (B1) The value of choosing the first constraint over nothing, (B2-4) Adding the first three scales up the demonstration likelihood drastically, (B5-7) after the fourth added constraint the continued scaling drops off. The fifth constraint happens to be a true constraint, but the sixth fails to identify the constraints outside the demonstrated area instead misselecting a constraint right on its boundary}
\end{figure*}

Algorithm 1 was implemented in MATLAB and tested on a synthetic dataset of 100 demonstrations.
This dataset was synthesized from simulated trajectories of a stochastically optimal agent minimizing distance traveled on a two-dimensional ``Gridworld'' MDP with movement in all eight compass directions.
These eight directions made up the action space $\mathcal{A}$ along with a loitering terminal action for once the goal was reached.
Each directional action was given a fixed ``slippage'' chance of $0.1$ where a random direction out of the other seven was followed instead.
All ground-truth and candidate state constraints were fixed at a constant chance threshold of $\psi = 0.25$ for all states.

The simulated demonstrator only noisily optimized the task, following a Boltzmann choice distribution as described in Equation (\ref{eq:BoltzmannDef}).
The constraint inference algorithm was evaluated on this dataset as shown in Figure \ref{fig:demonstration}. By the fifth iteration (shown in Figure \ref{fig:demonstration}e), the algorithm suceeded in recovering the groundtruth constraints (shown in Figure \ref{fig:demonstration}a).

\section{Limitations and Future Work}
The algorithms set forth in this paper focused on discretized state and action spaces.
For controlling many systems on practical timescales, the state must be handled as a continuous parameter.
Future work should investigate how gridded state spaces like in Figure \ref{fig:demonstration} could be refined to approximate continuous state spaces.
Reducing the algorithm to a variant Bellman backup, as we did in Theorem \ref{thm:dynprog}, suggests that the continuous variant may just be solving a Hamilton-Jacobi-Bellman equation.
These partial differential equations have a rich literature investigating their solution, including toolsets like \cite{MitchellToolbox}.

Extending constraint inference to stochastic systems raises questions of whether human experts might be better modeled using a prospect-theoretic or risk-sensitive measure as in \cite{mazumdar2017gradient}.
Future work should investigate how human heuristics for statistical prediction might impact the way demonstrations are generated.
The algorithm should be designed to be robust to these biases or even leverage their structure.

\section{Discussion and Conclusion}
By designing the likelihoods to maximize the causal entropy (that respects the information flow of state transition outcome revelation) this work makes maximum likelihood estimation possible for stochastic dynamics that reflect the uncertainties inherent in perilous situations.
Moreover, by broadening the hypothesis class to include chance constraints our algorithm not only learns the constraints from expert operators, but also their risk tolerances.
This opens the door to studying how expert operators plan risk-sensitively and what prospect-theoretic risk measures they may be employing.

Although increasing the complexity of systems that can be handled in constraint inference, this algorithm maintains the same computational complexity of $O(|\mathcal{X}|^2 (|\mathcal{X}| + |\mathcal{A}|)$ as prior art.
That is, control engineers can extract safety specifications from expert demonstration data for the same cost in both stochastic and deterministic dynamics.

\bibliographystyle{plain}
\bibliography{MLCI.bib}

\end{document}